\documentclass{llncs}
\usepackage[utf8]{inputenc}
\usepackage{graphicx, enumerate}
\usepackage{subfig}
\usepackage{amssymb}
\usepackage{amsmath}
\usepackage[english]{babel}
\usepackage{cite,hyperref}
\usepackage{xspace}
\usepackage{paralist}
\usepackage{times}
\usepackage{wrapfig}

\graphicspath{{fig/}}

\newcommand{\presence}{\Psi}
\newcommand{\activity}{\Phi}
\newcommand{\conflict}{C}

\newcommand{\T}{\ensuremath{\mathcal{T}}\xspace}

\newcommand{\R}{\mathbb{R}}
\newcommand{\GMT}{\textsc{GeneralMaxTotal}\xspace}

\newcommand{\GeneralMaxTotal}{\textsc{General\-Max\-Total}\xspace}
\newcommand{\kRestrictedMaxTotal}{$k$-\textsc{Re\-stric\-ted\-Max\-Total}\xspace}
\newcommand{\TwoRestrictedMaxTotal}{$2$-\textsc{Re\-stric\-ted\-Max\-Total}\xspace}
\newcommand{\GreedyMaxTotal}{\textsc{Greedy\-Max\-Total}\xspace}
\newcommand{\kRestrictedGreedy}{\textsc{Greedy\-Restricted\-Max\-Total
}\xspace}

\title{Trajectory-Based Dynamic Map Labeling} 
\author{Andreas Gemsa \and
  Benjamin Niedermann \and Martin~N\"ollenburg}
\institute{Karlsruhe Institute of Technology (KIT), Germany}

\begin{document}

\maketitle

\begin{abstract}

In this paper we introduce \emph{trajectory-based labeling}, a new variant of
dynamic map labeling, where a movement trajectory for the map viewport is
given. We define a general labeling model and study the active range
maximization problem in this model. The problem is {$\cal NP$}-complete and
$\mathcal W[1]$-hard. In the restricted, yet practically relevant case
that no more than $k$ labels can be active at any time, we give
polynomial-time algorithms. For the general case we present a
practical ILP formulation with an experimental evaluation as well as
approximation algorithms.

\end{abstract}

\section{Introduction}
In contrast to traditional static maps, dynamic digital maps support
continuous movement of the map viewport based on panning, rotation, or
zooming. 
Creating smooth visualizations under such map dynamics induces challenging geometric problems, e.g., continuous generalization~\cite{sb-cgvsm-04} or dynamic map labeling~\cite{bdy-dl-06}.
In this paper, we focus on map
labeling and take a trajectory-based view on it.  In many applications, e.g., car navigation, a movement
trajectory is known in advance and it becomes interesting to optimize
the visualization of the map locally along this trajectory.

Selecting and placing a maximum number of non-overlapping labels for
various map features is an important cartographic problem.  Labels are
usually modeled as rectangles and a typical objective in a static map is to
find a maximum (possibly weighted) independent set of labels.  This is
known to be {$\cal NP$}-complete~\cite{fpt-opcpn-81}. There are
several approximation algorithms and PTAS's in different labeling models~\cite{cc-mir-09,aks-lpmir-98}, as well as practically useful heuristics~\cite{ww-pla-95,wwks-trsglp-01}.

With the increasing popularity of interactive dynamic
maps, e.g., as %
digital globes or on mobile devices, the static labeling problem has
been translated into a dynamic setting.  Due to the temporal dimension
of the animations occurring during map movement, it is necessary to
define a notion of \emph{temporal consistency} or \emph{coherence} for map labeling as to avoid
distracting effects such as jumping or flickering
labels~\cite{bdy-dl-06}.  Previously, consistent labeling has been
studied from a global perspective under continuous
zooming~\cite{bnpw-oarcd-10} and continuous
rotation~\cite{gnr-clrm-11a}.  In practice, however, an individual map
user with a mobile device, e.g., a tourist or a car driver, is typically interested only in a specific part of a map and it
is thus often more important to optimize the labeling locally for a
certain trajectory of the map viewport than globally for the whole
map.

We introduce a versatile
trajectory-based model for dynamic map
labeling, and define three label activity models that guarantee concistency.  We apply
this model to point feature labeling for a viewport that moves and
rotates along a differentiable trajectory in a fixed-scale base map in
a forward-facing way.  Although we present our approach in a very specific
problem setting, our model is very general. Our approach can
be applied for every dynamic labeling problem that can be expressed as
a set of label availability intervals over time and a set of conflict
intervals over time for pairs of labels. The exact algorithms hold for
the general model, the approximation algorithm itself is also
applicable, but the analysis of the  approximation ratio requires
problem-specific geometric arguments, which must be adjusted to the specific setting. 

\noindent \textbf{Contribution.} For our specific problem, we show that maximizing the number of
visible labels integrated over time in our model is {$\cal
  NP$}-complete; in fact it is even $\mathcal W[1]$-hard and thus it
is unlikely that a fixed-parameter tractable algorithm exists. We
present an  integer linear programming (ILP) formulation for
the general unrestricted case, which is supported by a short
experimental evaluation. For the special case of unit-square labels we
give an efficient approximation algorithm with different
approximation ratios
depending on the actual label activity
model. Moreover, we present polynomial-time algorithms for the
restricted case that no more than $k$ labels are active at any time
for some constant $k$. We note that limiting the number of
simultaneously active labels is of practical interest as to avoid
overly
dense labelings, in particular for dynamic maps on small-screen
devices such as in car navigation systems.

\section{Trajectory-Based Labeling Model}\label{sec:model}
Let $M$ be a labeled north-facing, fixed-scale map, i.e., a set of
points $P=\{p_1,\dots,p_N\}$ in the plane together with a
corresponding set $L=\{\ell_1,\dots,\ell_N\}$ of labels. Each label
$\ell_i$ is represented by an axis-aligned rectangle of individual
width and height. We call the point~$p_i$ the \emph{anchor} of the
label~$\ell_i$. Here we assume that each label has an arbitrary but
fixed position relative to its anchor, e.g., with its lower left
corner coinciding with the anchor.  The \emph{viewport}~$R$ is an
arbitrarily oriented rectangle of fixed size that defines the
currently visible part of~$M$ on the map screen.  The viewport follows
a trajectory that is given by a continuous differentiable function
$T\colon [0, 1] \to \R^2$.  For an example see
Fig.~\ref{fig:example_trajectory}.
\begin{figure}[tb]
 \centering
 \includegraphics[page=2, scale=1]{./fig/example_trajectory} \hfil
 \includegraphics[page=1, scale=1]{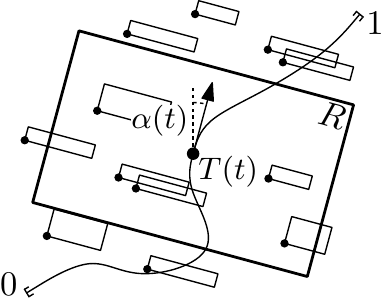}
 \caption{\small Illustration of the viewport moving along a trajectory.
   Left the user's view and right a general view of the map and the
   viewport.}
 \label{fig:example_trajectory}
\end{figure}
More precisely, we describe the viewport by a function~$V\colon
[0,1]\to \mathbb{R}^2\times [0,2\pi]$. The interpretation of
$V(t)=(c,\alpha)$ is that at time~$t$ the center of the rectangle~$R$
is located at~$c$ and~$R$ is rotated clockwise by the
angle~$\alpha$ relatively to a north base line of the map.
Since~$R$ moves along~$T$ we define~$V(t)=(T(t),\alpha(t))$,
where~$\alpha(t)$ denotes the direction of~$T$ at time~$t$.  For
simplicity, we sometimes refer to~$R$ at time~$t$ as~$V(t)$.  To
ensure good readability, we require that the labels are always aligned
with the viewport axes as the viewport changes its orientation, i.e.,
they rotate around their anchors by the same angle~$\alpha(t)$, see
Fig.~\ref{fig:example_trajectory}.  We denote the rotated label
rectangle of~$\ell$ at time~$t$ by~$\ell(t)$.

We say that a label~$\ell$ is \emph{present} at time~$t$, if~$V(t)\cap
\ell(t)\neq \emptyset$. As we consider the rectangles~$\ell(t)$
and~$V(t)$ to be closed, we can describe the points in time for
which~$\ell$ is present by closed intervals.  We define for each
label~$\ell$ the set~$\presence_\ell$ that describes all disjoint
subintervals of~$[0,1]$ for which~$\ell$ is present, thus
$\presence_\ell=\{[a,b] \mid [a,b]\subseteq[0,1]$ is maximal so that
$\ell$ is present at all $t\in[a,b]\}$.  Further, we define the
disjoint union~$\presence=\{([a,b],\ell) \mid [a,b]\in \presence_\ell$
and $\ell \in L\}$ of all $\presence_\ell$. We
abbreviate~$([a,b],\ell)\in\presence$ by~$[a,b]_\ell$ and call
$[a,b]_\ell\in \presence$ a \emph{presence interval} of $\ell$. In the
remainder of this paper we denote the number of presence intervals by
$n$.

Two labels~$\ell$ and~$\ell'$ are in \emph{conflict} with each other
at time~$t$ if $\ell(t)\cap\ell'(t)\neq \emptyset$.
If~$\ell(t)\cap\ell'(t)\cap V(t)\neq \emptyset$ we say that the
conflict is \emph{present} at time $t$. As in~\cite{gnr-clrm-11a} we
can describe
the occurrences of conflicts between two labels~$\ell,\ell'\in L$ by a
set of closed intervals:~$\conflict_{\ell,\ell'}=\{[a,b]\subseteq
[0,1]\mid [a,b]$ is maximal and~$\ell$ and $\ell'$ are in conflict at
all $t\in[a,b]\}$.  We define the disjoint
union~$C=\{([a,b],\ell,\ell')\mid [a,b]\in C_{\ell,\ell'}$ and
$\ell,\ell'\in L\}$ of all $\conflict_{\ell,\ell'}$.  We
abbreviate~$([a, b], \ell, \ell') \in C$ as $[a, b]_{\ell, \ell'}$ and
call it a \emph{conflict interval} of~$\ell$ and~$\ell'$.  Two
presence intervals~$[a,b]_\ell$ and~$[c,d]_{\ell'}$ \emph{are in
  conflict} if there is a conflict~$[f,g]_{\ell,\ell'}\in\conflict$
s.t.\ the intersection of the
intervals~$[f,g]_{\ell,\ell'}\cap[a,b]_\ell\cap[c,d]_{\ell'}
\neq\emptyset$.

The tuple~$(P,L,\presence,\conflict)$ is called an \emph{instance} of
trajectory-based labeling. Note that the essential information of~$T$
is implicitly given by $\presence$ and $\conflict$ and that for each
label~$\ell\in L$ there can be several presence intervals. In this paper we assume that $\presence$ and $\conflict$ is given as input.
In practice, however, we usually first need to compute~$\presence$ and~$\conflict$ given a continuous and differentiable trajectory $T$. An interesting special case is that~$T$ is a continuous, differentiable chain of~$m$ circular arcs (possibly of infinite radius), e.g., obtained by approximating a polygonal route in a road network. Niedermann~\cite{n-cldmust-12} showed that in this case the set $\presence$ can be computed in~$O(m\cdot N)$ time
and the set~$C$ in~$O(m\cdot N^2)$ time. His main observation was
that for each arc of $T$ the viewport can in fact be treated as a huge
label and that ``conflicts'' with the viewport correspond to presence
intervals. We refer to~\cite[Chapter 15]{n-cldmust-12} for
details.

Next we define the \emph{activity} of labels, i.e., when to actually
display which of the present labels on screen.  We restrict ourselves
to closed and disjoint intervals describing the activity of a
label~$\ell$ and define the set~$\activity_\ell=\{[a,b]\subseteq[0,1]
\mid [a,b]$ is maximal such that $\ell$ is active at all $t\in[a,b]\}$, as well as the disjoint
union~$\activity=\{([a,b],\ell)\mid [a,b]\in \activity_\ell$ and
$\ell\in L\}$ of all $\activity_\ell$. We
abbreviate~$([a,b],\ell)\in\activity$ with~$[a,b]_\ell$ and call
$[a,b]_\ell\in \activity$ an \emph{active interval} of~$\ell$.

It remains to define an \emph{activity model} restricting $\activity$
in order to obtain a reasonable labeling. Here we propose
three activity models AM1, AM2, AM3 with increasing flexibility.  All
three activity models exclude overlaps of displayed labels and
guarantee consistency criteria introduced by Been et
al.\cite{bdy-dl-06}, i.e., labels must not flicker or jump. To that
end they share the following
 properties \textbf{(A)} a label~$\ell$ can only be active at time $t$
if it is present at time $t$, \textbf{(B)} to avoid flickering and
jumping each presence interval of~$\ell$ contains at most one active
interval of~$\ell$, and \textbf{(C)} if two labels are in conflict at
a time~$t$, then at most one of them may be active at 
$t$ to avoid overlapping labels.

What distinguishes the three models are the possible points in time
when labels can become active or inactive.  The first and most
restrictive activity model AM1 demands that each activity
interval~$[a,b]_\ell$ of a label~$\ell$ must coincide with a presence
interval of~$\ell$. The second activity model AM2 allows an active
interval of a label $\ell$ to end earlier than the corresponding
presence interval if there is a \emph{witness label} $\ell'$ for that,
i.e., an active interval for~$\ell$ may end at time~$c$ if there is a
starting conflict interval $[c, d]_{\ell, \ell'}$ and the conflicting
label $\ell'$ is
active at $c$. However, AM2 still requires every active interval to
begin with the corresponding presence interval. The third activity
model AM3 extends AM2 by also relaxing the restriction regarding the
start of active intervals.  An active interval for a label $\ell$ may
start at time $c$ if a present conflict $[a,c]_{\ell,\ell'}$ involving
$\ell$ and an active witness label $\ell'$ ends at time $c$. In this
model active intervals may begin later and end earlier than their
corresponding presence intervals if there is a visible reason for the
map user to do so, namely the start or end of a conflict with an
active witness label.

A common objective in both static and dynamic map labeling is to
maximize the number of labeled points. Often, however, certain labels
are more important than others. To account for this, each label $\ell$
can be assigned a weight $W_\ell$ that corresponds to its
significance. Then we define the weight of an interval $[a, b]_\ell
\in \activity$ as $w([a, b]_\ell) = (b-a) \cdot W_\ell$.  
Given an
instance~$(P,L,\presence,\conflict)$, then with respect to one of the
three activity models we want to find an activity $\Phi$ that
maximizes~$\sum_{[a,b]_\ell\in \activity}w([a,b]_\ell)$; we call this
optimization problem \textsc{GeneralMaxTotal}.  If we require
that at
any time $t$ at most~$k$ labels are active for some~$k$, we
call the problem
\mbox{$k$-\textsc{RestrictedMaxTotal}}. In particular the
latter problem is interesting for small-screen devices, e.g., car
navigation systems, that should not overwhelm the user with
additional information.

\section{Solving \GeneralMaxTotal}

We first prove that
\GMT is $\mathcal NP$-complete. The membership of
\GMT in~$\mathcal NP$ follows from the fact that
the start and the end of an active interval must coincide with the start or
end of a presence interval or a conflict interval. Thus, there is a
finite number of candidates for the endpoints of the active intervals so
that a solution~$\mathcal L$ can be guessed. Verifying
that~$\mathcal L$ is valid in one of the three models and that its value
exceeds a given threshold can obviously be checked in polynomial time.

For the $\mathcal NP$-hardness we apply a straight-forward reduction
from the $\cal NP$-complete maximum independent set of rectangles
problem~\cite{fpt-opcpn-81}. We simply interpret the set of rectangles
as a set of labels with unit weight, choose a short vertical
trajectory $T$ and a viewport $R$ that contains all labels at any
point of $T$. Since the conflicts do no change over time, the
reduction can be used for all three activity models. By means of
the same reduction and Marx'
result~\cite{Marx05} that finding an independent set for a given set
of axis-parallel unit squares is $\mathcal W[1]$-hard we derive the
next theorem.

\begin{theorem}
  \label{thm:gmt:npc}
  \GMT is $\mathcal{NP}$-complete and
  $\mathcal{W}[1]$-hard for all activity models AM1--AM3.
\end{theorem}

As a consequence, \GMT is not fixed-parameter tractable
unless $\mathcal{W}[1] =\mathcal{FPT}$. Note that this also means that
for $k$-\text{RestrictedMaxTotal} we cannot expect to find an
algorithm that runs in $O(p(n)\cdot C(k))$ time, where $p(n)$ is a
polynomial that depends only on the number~$n$ of presence intervals
and the computable function~$C(k)$ depends only on the parameter~$k$.

\subsection{Integer Linear Programming for \GeneralMaxTotal}
Since we are still interested in finding an optimal solution for
\GMT we have developed integer linear programming
(ILP) formulations for all three activity models. We present the formulation for the most involved model AM3 and then argue how to adapt it to the simpler models AM1 and AM2.

\begin{wrapfigure}[10]{r}{.45\textwidth}
    \centering
     \vspace{-3ex}
     \includegraphics{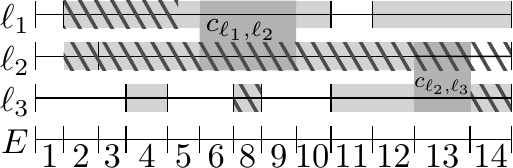}
     \caption{\small 
       Depiction of presence intervals (light gray), active intervals (hatched), and conflicts (dark gray).
     }
   \label{fig:example_intervals}
   \end{wrapfigure}
We define $E$ to be the totally ordered set of the endpoints
of all presence and all conflict intervals and include 0
and 1; see~Fig.~\ref{fig:example_intervals}. We call each
interval~$[c,d]$ between two consecutive elements~$c$ and~$d$ in~$E$
an \emph{atomic segment} and denote the~$i$-th atomic segment of~$E$
by~$E(i)$. Further, let $X(\ell,i)$ be the set of labels that are in
conflict with~$\ell$ during $E(i-1)$, but not during $E(i)$, i.e.,
the conflicts end with~$E(i-1)$. Analogously, let $Y(\ell,i)$ be the
set of labels that are in conflict with~$\ell$ during $E(i+1)$, but
not during $E(i)$, i.e., the conflicts begin with~$E(i+1)$.
For each label~$\ell$ we introduce
three binary variables $b_i,x_i,e_i\in\{0,1\}$ and the following constraints.
{\small \begin{align}
\label{eq:A}b_i^\ell =x_i^\ell =e_i^\ell=0 && \forall 1\leq i\leq
|E|\ \text{s.t. }\forall [c,d]\in
\presence_\ell: E(i)\cap [c,d]=\emptyset \\
\label{eq:B}\sum_{j \in J} b^\ell_{j}\leq 1 \text{ and } \sum_{j \in
J}
e^\ell_{j}\leq 1 && \forall [c,d]\in
\presence_\ell \text{ where }J=\{j \mid E(j) \subseteq [c,d]\}\\
\label{eq:C}x^\ell_i+x^{\ell'}_i\leq 1 && \forall 1\leq i\leq
|E|\ \forall
[c,d]_{\ell,\ell'}\in C:\ E(i)\subseteq[c,d]\\
\label{eq:D}x^\ell_{i-1} + b^\ell_i = x^\ell_{i} + e^\ell_{i-1} &&
\forall 1\leq i\leq
|E|\ (\text{set } x_0=e_0=0)\\
\label{eq:E}b^\ell_{j}\leq \sum_{\ell'\in X(\ell,j)}
x^{\ell'}_{j-1}  && \forall [c,d]_\ell \in\presence\ \forall
E(j) \subset [c,d]_\ell \text{ with } c\not\in E(j)\\
\label{eq:F}e^\ell_{j}\leq \sum_{\ell'\in Y(\ell,j)} x^{\ell'}_{j+1} 
&& \forall [c,d]_\ell \in\presence\ \forall
E(j) \subset [c,d]_\ell \text{ with } d\not\in E(j)
\end{align}}
Subject to these constraints we maximize $\sum_{\ell \in
L}\sum_{i=1}^{|E|-1} x_i^{\ell}\cdot w(E(i))$.
The intended meaning of the variables is that $x^\ell_i = 1$
if $\ell$ is active during $E(i)$ and otherwise
$x^\ell_i = 0$. Variable $b^\ell_i = 1$ if and only
if~$E(i)$ is the first atomic segment of an active interval of $\ell$,
and analogously~$e^\ell_i=1$ if and only if~$E(i)$ is the last atomic
segment of an active interval of~$\ell$.
Recall the properties of the activity models as defined in
Section~\ref{sec:model}. Constraints~(\ref{eq:A})--(\ref{eq:C})
immediately ensure properties~(A)--(C), respectively. 
Constraint~(\ref{eq:D}) means that if~$\ell$ is active during~$E(i-1)$
($x^\ell_{i-1}=1$),  then it must either
 stay active during~$E(i)$
($x^\ell_{i}=1$) or the active interval ends with~$E(i-1)$
($e^\ell_{i-1}=1$), and if~$\ell$ is active during~$E(i)$
($x^\ell_{i}=1$) then it must be active during~$E(i-1)$
($x^\ell_{i-1}=1$) or the active interval begins with~$E(i)$
($b^\ell_{i}=1$).  Constraint~(\ref{eq:E}) enforces that for~$\ell$ to
become active with~$E(j)$ at least one witness label of
$X(\ell,j)$ is active during~$E(j-1)$. Analogously,
constraint~(\ref{eq:F}) enforces that for~$\ell$ to
become inactive with~$E(j)$ at least one witness label of
$Y(\ell,j)$ is active during~$E(j+1)$. Note that without the explicit
constraints (\ref{eq:E}) and (\ref{eq:F}) two conflicting labels could
switch activity at any point during the conflict interval rather than
only at the endpoints. For an
example see~Fig.~\ref{fig:ilp:constraint:necessary}.  The drawing
shows an optimal solution that is valid for the
    ILP formulation if the constraints (5) and (6) are omitted. In
    particular~$\ell_1$ becomes inactive at time~$t$, although~$t$ is
    not the right boundary of the corresponding presence interval and
    there is no conflict of~$\ell_1$ that begins at~$t$ such that the
    corresponding opponent is active from~$t$ on. Analogous
    observations can be made for~$\ell_2$. Consequently, this solution
    does not satisfy AM3.

\begin{figure}[tb]
  \centering
  \includegraphics[page=3]{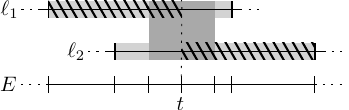}
  \caption{\small The light gray intervals show presence intervals,
the
    hatched intervals active intervals and the dark gray intervals
    conflicts between labels.  Further, the bottom line illustrates
    possible atomic segments, when assuming that there is a third
    label that induces the segmentation at
    time~$t$.
     }
\label{fig:ilp:constraint:necessary}
\end{figure}

\newcommand{\thmilptext}{
Given an instance $I=(P,L,\presence,C)$, the
  ILP (\ref{eq:A})--(\ref{eq:F}) computes an optimal solution~$\activity$ of
  \GMT in AM3. It uses $O(N\cdot (|\presence| + |C|))$
  variables and constraints.}

\begin{theorem}\label{thm:ilp}
\thmilptext
\end{theorem}

\begin{proof}
  Every solution of the ILP corresponds to an activity~$\activity$ by
  defining for every label $\ell$ the set $\activity_\ell$ as the set
  of all maximal intervals in $\bigcup_{i:x_i^\ell=1}
  E(i)$. Conversely, every valid activity~$\activity$ in AM3 can be
  expressed in terms of the variables of the ILP. To show that we
  first observe that for every valid activity interval $[a,b]_{\ell}$
in
  AM3 the endpoints $a$ and $b$ are necessarily endpoints of a
  conflict interval or a presence interval of $\ell$. Thus
  $[a,b]_\ell$ can be expressed as the union of consecutive atomic
  segments represented by the variables $x_i^\ell$.
  
  It is clear that the objective function computes the weight of a
  solution $\activity$ correctly. Thus it remains to show that the
  constraints (1)--(6) indeed model AM3, i.e., every solution of the
  ILP satisfies AM3 and every activity in AM3 is a solution of the
  ILP.  It follows immediately from the definition of constraints
  (1)--(3) that they model properties (A)--(C), assuming that the
  start- and endpoint of every activity interval is indeed marked by
  setting $b_i^\ell=1$ and $e_j^\ell=1$ for its first and last atomic
  segments $E(i)$ and $E(j)$. But this is achieved by the
  constraints~(4) as discussed above.  Now in AM3 a label can only
  become active (inactive) at the start (end) of its presence interval
  or at the end (start) of a conflict interval if the conflicting
  label is active as a witness. We show that constraint (5) yields
  that the start of an activity interval is correct according to
  AM3. The argument for the end of an activity interval follows
  analogously from constraint (6).  Let $E(i)$ be the first atomic
  segment in an activity interval of the label $\ell$. Then by
  constraint (4) we have $b_i^\ell = 1$ and $x_i^\ell = 1$. If $E(i)$
  is the first segment of a presence interval then this is a valid
  start according to AM3. Note that constraint (5) is not
  present in that case and thus does not restrict $b_i^\ell$.
Otherwise let $E(i)$ be not the first
  segment of a presence interval.  Then for this segment the ILP
  contains constraint (5). If no conflict interval of $\ell$ ends with
  $E(i-1)$ then (5) sets $b_i^\ell = 0$ anyways, so this is not
  possible. If some conflict intervals of $\ell$ end with $E(i-1)$ but
  none of them are active in $E(i-1)$ then constraint (5) also yields
  $b_i^\ell = 0$. So the only two possibilities for $b_i^\ell = 1$ are
  that either $E(i)$ is the first segment of a presence interval or
  $E(i)$ is the first segment after a conflict interval of $\ell$ for
  which a witness label is active. Thus every solution of the ILP
  satisfies AM3.

  Conversely, let $\activity$ be valid according to AM3. Since
  $\activity$ satisfies properties (A)--(C), the corresponding
  assignment of binary values to the variables $x_i^\ell$, $b_i^\ell$,
  and $e_i^\ell$ satisfy constraints (1)--(4). It remains to show that
  the constraints (5) and (6) hold. Let $[a,b]_\ell \in \activity$ be
  a particular activity interval and let $E(i)$ be the atomic segment
  starting at $a$.  If $a$ is the start of a presence interval of
  $\ell$ then there is no constraint (5) for $\ell$ and the segment
  $E(i)$ and thus it is possible to have $b_i^\ell = 1$. Otherwise,
  $a$ is the end of a conflict interval of $\ell$ with another label
  $\ell'$ that is an active witness in the atomic segment $E(i-1)$
  ending at $a$. This means that $x_{i-1}^{\ell'} = 1$ and thus
  constraint (5) is satisfied for $b_i^\ell = 1$. Analogous reasoning
  for the endpoints of all activity intervals and constraint (6) yield
  that $\activity$ can indeed be represented as a solution to the ILP.

  Since the number of atomic segments is $O(|\presence| + |C|)$ and
  there are $N$ labels the bound on the size of the ILP follows.  \qed
\end{proof}

We can adapt the above ILP to AM1 and AM2 as follows. For AM2 we
replace the right hand side of constraint~(\ref{eq:E}) by $0$, and for
AM1 we also replace the right hand side of constraint~(\ref{eq:F}) by
$0$. This
excludes exactly the start- and endpoints of the activity intervals
that are forbidden in AM1 or AM2. It is easy to see that these ILP
formulations can be modified further to solve \kRestrictedMaxTotal by
adding the constraint $\sum_{\ell \in L} x_i^\ell \le k$ for each
atomic segment $E(i)$.
\begin{corollary}\label{cor:ilp}
  Given an instance $I=(P,L,\presence,C)$, \GMT and
  \kRestrictedMaxTotal can be solved in AM1, AM2, and AM3 by an ILP
  that uses $O(N\cdot (|\presence| + |C|))$ variables and constraints.
\end{corollary}

\subsection{Experiments.}

We have evaluated the ILP in all three models using Open Street Map
data of the city center of Karlsruhe (Germany) which contains more
than 2,000 labels.  To this end we generated 1,000 shortest paths on
the road network of Karlsruhe by selecting source and target vertices
uniformly at random and transformed those shortest paths into
trajectories consisting of circular arcs.  We fixed the viewport's
size to that of a typical mobile device ($640\times 480$ pixels) and
considered the map scales 1:2000, 1:3000, and~1:4000, which
corresponds to areas with dimensions~$339m\times
254m$,~$508m\times381m$ and~$678m\times 508m$, respectively.  The
experiments were performed on a single core of an AMD Opteron 6172
processor running Linux~3.4.11. The machine is clocked at 2.1 Ghz, and
has 256 GiB RAM. Our implementation is written in C++, uses Gurobi
5.1. as ILP solver, and was compiled with GCC 4.7.1 using optimization
\texttt{-O3}.

For plots and a table depicting the results of the experimental
evaluation see Fig.~\ref{fig:evaluation}. We observe that for a scale
factor of~1:2000 the running times for the vast majority of instances
remained below one second, while no instances required more than ten
seconds to be solved.  Since for the scale factors~1:3000 and~1:4000
the density of labels increases, the running times increase,
too. Still 75\% of the instances were solved in less than three
seconds.

It is remarkable that for the scale factor~1:2000 over~99\% of the
instances and for the scale factor~1:3000 over 75\% of the instances
can be solved in less than a second, while for the scale factor~1:4000
over~75\% of the instances can still be solved in less than three
seconds.  However, for~1:3000 there are runs that needed almost~50
seconds and for~1:4000 there are runs that needed almost~475 seconds.
Note that due to these outliers, for a scale factor of~1:4000 the
average running time lies above the third quartile, but still does not
exceed six seconds.  Two instances for 1:4000 exceeded a timeout of
600 seconds, were aborted and not included in the analysis.  For the
scale factors~1:2000 and~1:3000 the average running time is less than
one second. Considering the same scale factor the three models do not
differ much from each other, except for some outliers.  As the number
of labels and conflicts to be considered depends on the applied scale
factor and the concrete trajectory the table summarizes the number of
considered labels and conflicts in maximum and average over all
trajectories. As to be expected for a scale factor of~1:4000 the
number of considered conflicts is significantly greater than for a
scale factor of~1:2000. This also explains the different running
times.

In conclusion, our brief evaluation indicates that the ILP
formulations are indeed applicable in practice.

\begin{figure}
 \centering 
 \subfloat[AM1]{\includegraphics[width=.45\textwidth]{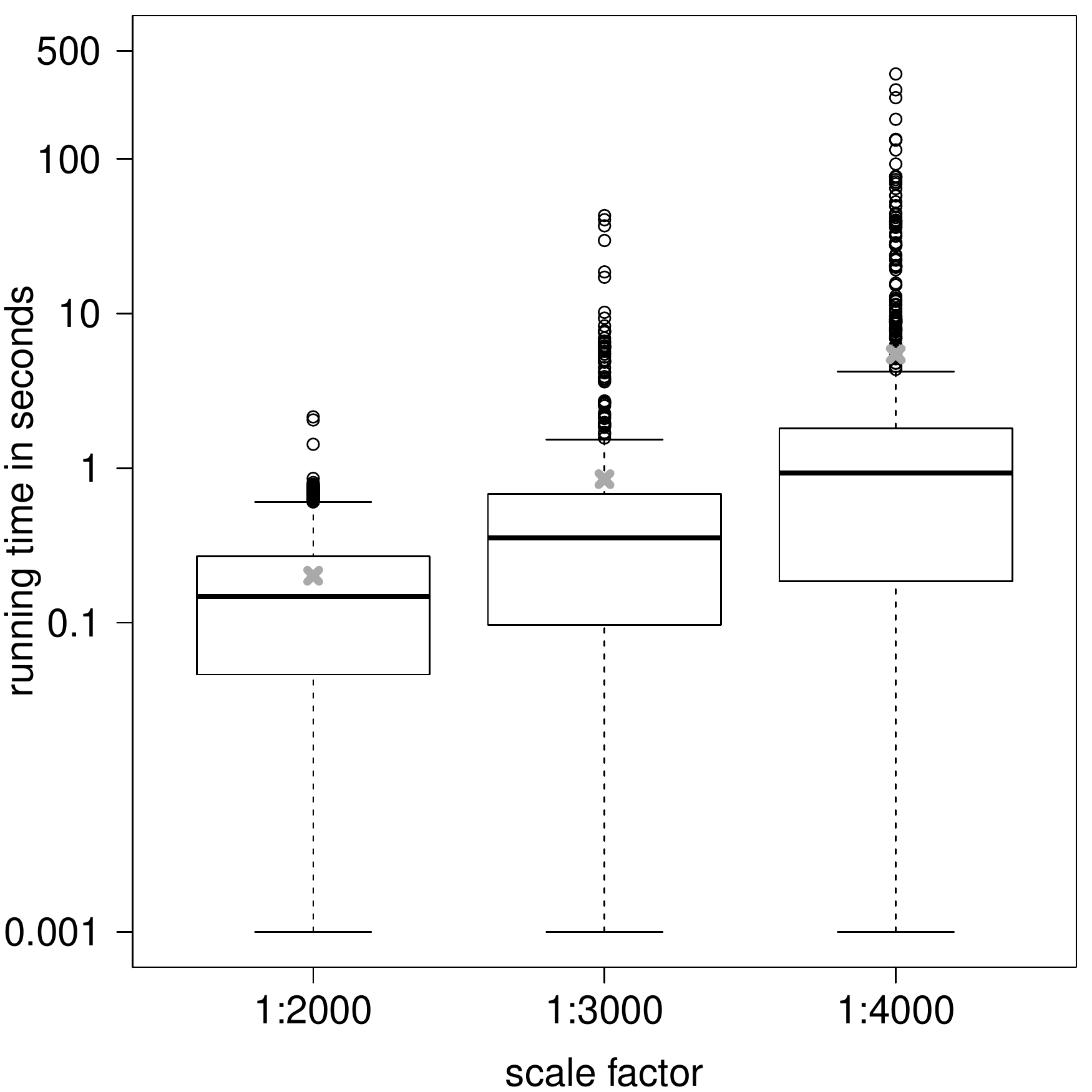}}
\hfill 
\subfloat[AM2]{\includegraphics[width=.45\textwidth]{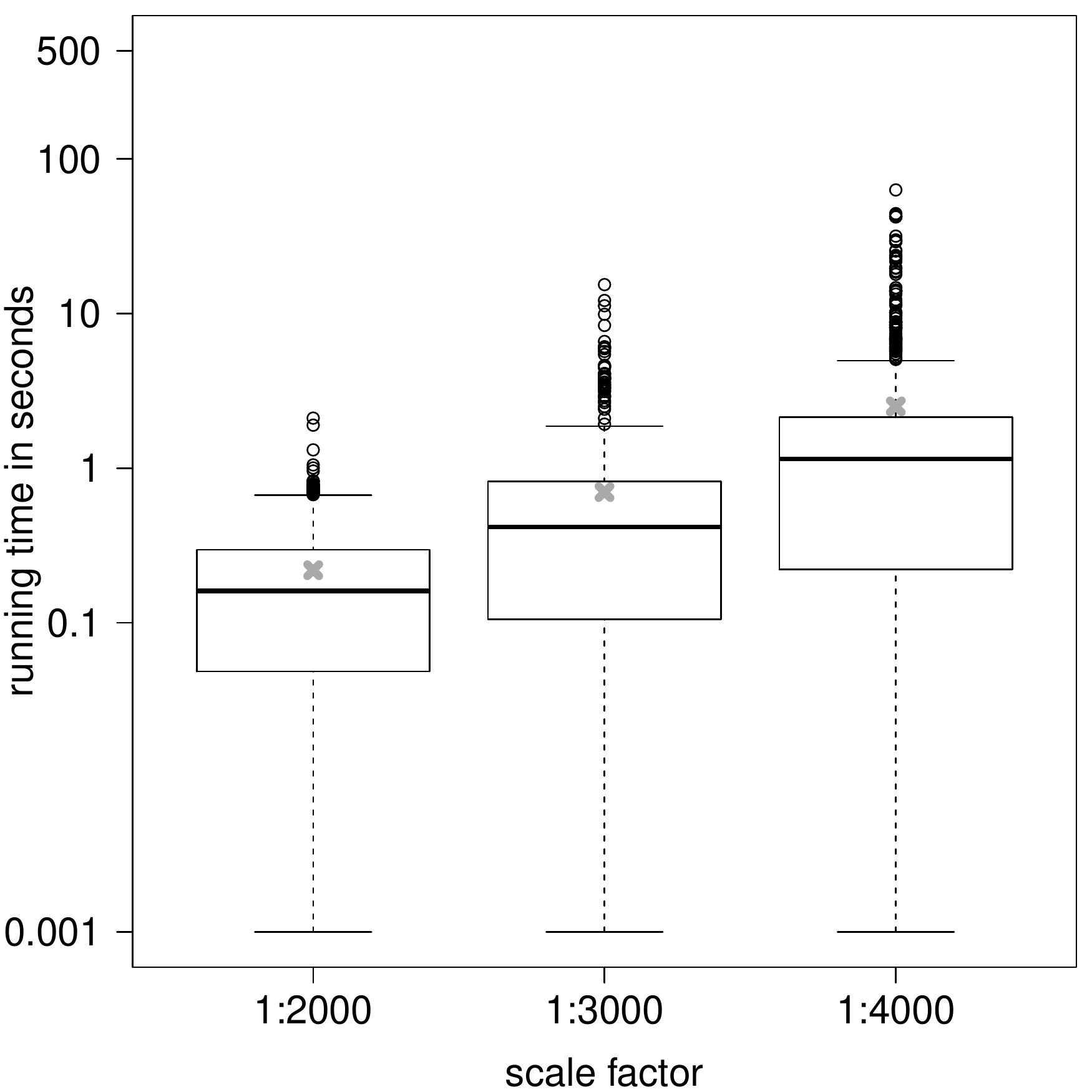}}\\
\subfloat[AM3]{\includegraphics[width=.45\textwidth]{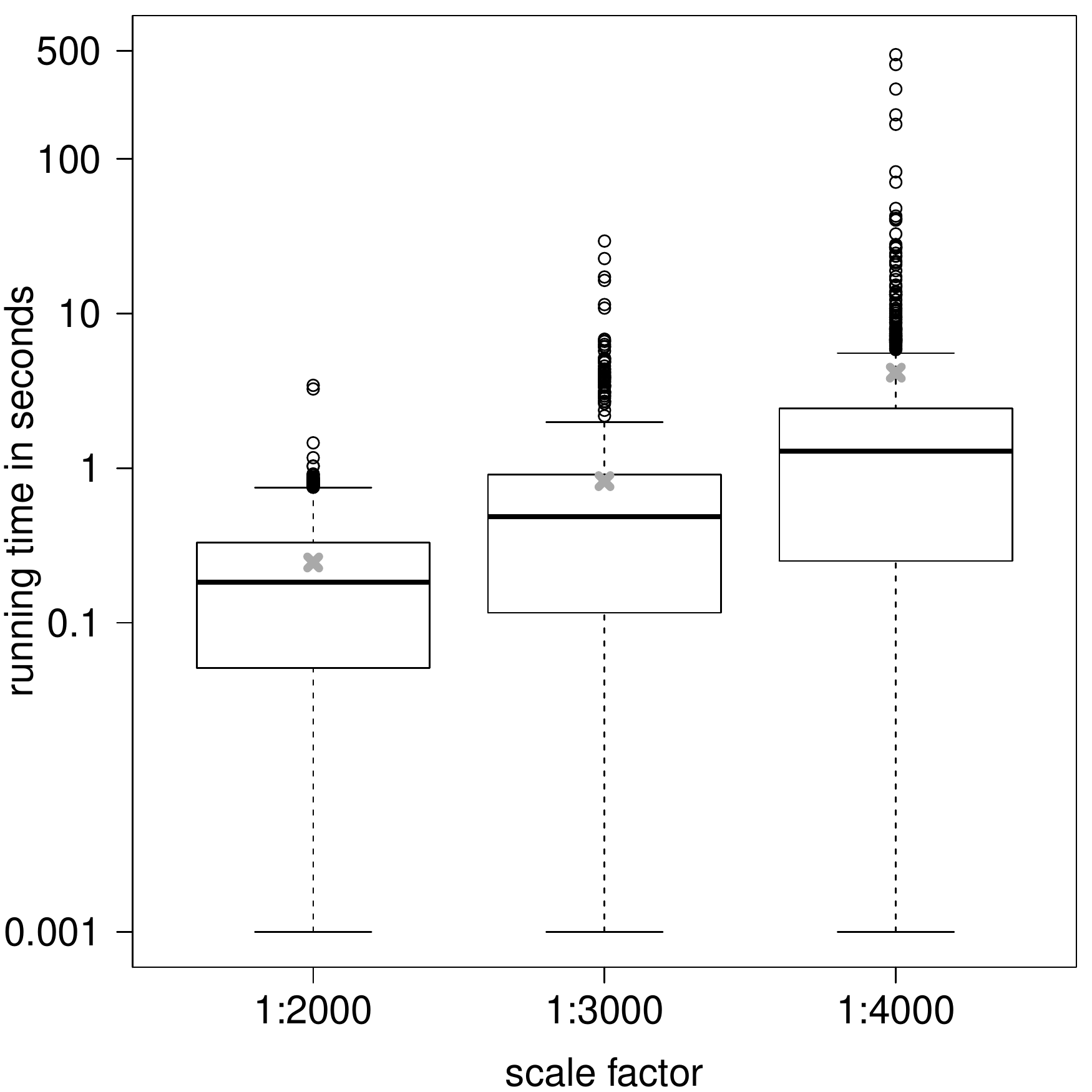}}
\hfill 
\subfloat[Statistics for $1,000$
trajectories]{\includegraphics{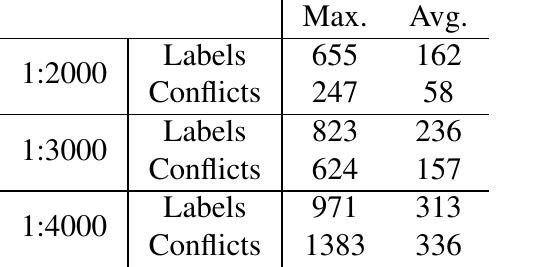}}

 \caption{\small Results of the experimental evaluation. In order to
limit the vertical axis in the plots we
rounded up all running times below 1ms to 1ms. 
 }
\label{fig:evaluation}
\end{figure}

\subsection{Approximation of \GeneralMaxTotal}
In this section we describe a simple greedy algorithm for
\GMT in all three activity models assuming that all labels
are unit squares anchored at their lower-left corner. Further, we
assume that the weight of each presence interval $[a, b]_\ell$ is its
length $w([a, b]_\ell) = b-a$.

Starting with an empty solution~$\activity$, our algorithm
\GreedyMaxTotal removes the longest interval~$I$ from~$\presence$ and
adds it to~$\activity$, i.e.,~$I$ is set active. Then, depending on
the activity model, it updates all presence intervals that have a
conflict with $I$ in $\Psi$ and continues until the set $\presence$ is
empty.

For AM1 the update process simply removes all presence intervals from
$\Psi$ that are in conflict with the newly selected interval $I$.
For AM2 and AM3 let~$I_j\in \presence$ and let~$I^1_j,\ldots,I^k_j$ be
the longest disjoint sub-intervals of~$I_j$ that are not in conflict
with the selected interval $I$. We assume that~$I^1_j,\ldots,I^k_j$
are sorted by their left endpoint. The update operation for AM2
replaces every interval $I_j \in \presence$ that is in conflict
with~$I$ with~$I^1_j$. In AM3 we replace~$I_j$ by~$I_j^1$, if~$I_j^1$
is not fully contained in~$I$. Otherwise, $I_j$ is replaced
by~$I^k_j$. Note that this discards some candidate intervals, but the
chosen replacement of $I_j$ is enough to prove the approximation
factor.  Note that after each update all intervals in $\Psi$ are valid
choices according to the specific model. Hence, we can conclude that
the result~$\Phi$ of \GreedyMaxTotal is also valid in that model.

In the following we analyze the approximation quality
of~\GreedyMaxTotal. To that end we first introduce a purely geometric
packing lemma. Similar packing lemmas have been introduced before,
but to the best of our knowledge for none of them it is
sufficient that only one prescribed corner of the packed objects lies
within the container.

\newcommand{\lempackingsquarestext}{
Let~$C$ be a circle of radius~$\sqrt 2$ in the plane and let
$Q$ be a set of non-intersecting closed and axis-parallel unit squares
with their bottom-left corner in~$C$. Then~$Q$ cannot contain more
than eight squares.
}

\begin{lemma}\label{lem:packing:squares}
\lempackingsquarestext
\end{lemma}

\begin{proof}
  First, we show that~$Q$ cannot contain more than nine squares and
  extend the result to the claim of the lemma. We begin by proving the
  following claim.

  \emph{(S) At most three squares of~$Q$ can be stabbed by a vertical
    line.}  In order to prove (S) let~$Q'\subseteq Q$ be a set of
  squares that is stabbed by an arbitrary vertical line~$l$ and
  let~$q_t$ be the topmost square stabbed by~$l$ and let~$q_b$ be the
  bottommost square stabbed by~$l$. Since both the bottom-left corner
  of $q_t$ and $q_b$ are in $C$, their vertical distance is at
  most~$2 \sqrt{2}$. Consequently, there can be at most one other
  square in~$Q'$ that lies in between~$q_t$ and~$q_b$, which shows the
claim (S).
 \begin{figure}[tb]
      \centering
    \includegraphics{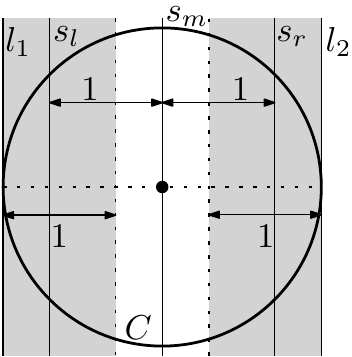}
    \caption{\small Illustration for the proof of
      Lemma~\ref{lem:packing:squares}.}
  \label{fig:packing:squares}
 \end{figure}

  Now let $l_1$ be the left vertical tangent of $C$ and let 
  $l_2$ be its right vertical tangent;
  see Fig.~\ref{fig:packing:squares}. We define~$Q_l \subseteq Q$ to
  be the set of squares whose bottom-left corner has distance of at
  most 1 to~$l_1$.  Hence, there must be a vertical line that stabs
  all squares in~$Q_l$. By (S) it follows that~$|Q_l| \leq 3$.  We can
  analogously define the set $Q_r \subseteq Q$ whose bottom-left
corner has
  distance of at most one to the vertical line~$l_2$. By the same
  argument it follows that $|Q_r| \leq 3$.
  Further, the bottom-left corners of the
  squares~$Q_m=Q\setminus\{Q_l,Q_r\}$ must be contained in a vertical
  strip of width~$2\sqrt{2}-2 < 1$.  Hence, there is a vertical
  line that stabs all squares of~$Q_m$ and $|Q_m|
  \leq 3$ follows.  We conclude that the set $Q$ contains at most nine
  squares; in fact, $|Q| \le 8$ as we show next.

  For the sake of contradiction we assume that $|Q|=9$, i.e.,
  $|Q_l|=|Q_m|=|Q_r|=3$. We denote the topmost square in $Q_l$ by
  $t_l$ and the bottommost square by $b_l$, and define $t_r$ and $b_r$
  for $Q_r$ analogously.
  Further, let~$s_m$ be the vertical line through the center of~$C$,
  let~$s_l$ be the vertical line that lies one unit to the left
  of~$s_m$ and
  let~$s_r$ be the vertical line that lies one unit to the right
  of~$s_m$. Note that the length of the segment of $s_l$ and $s_r$
  that is contained in $C$ has length~2. Since the bottom-left corners
  of~$t_l$ and~$b_l$ must have vertical distance strictly greater than
2,
  both squares must lie to the right of~$s_l$. Hence, $t_l$ and~$b_l$
  intersect~$s_m$. Analogously, the bottom-left corners of~$t_r$
  and~$b_r$ must lie to the left of~$s_r$, and, hence intersect~$s_r$.
  The line~$s_m$ is intersected by two squares of~$Q_l$. By (S) there
  can be at most one additional square of~$Q_m$ that
  intersects~$s_m$. Thus, there must be two squares in $Q_m$ whose
  anchors lie to the right of~$s_m$. But then they both
  intersect~$s_r$ which itself is already intersected by at least the
  squares $t_r$ and $b_r$. This is a contradiction to (S), and
concludes the
  proof. \qed
\end{proof}

Fig.~\ref{fig:eight_squares:example} shows
that the bound %
is tight.
\begin{figure}[tb]
  \centering
  \includegraphics[page=2]{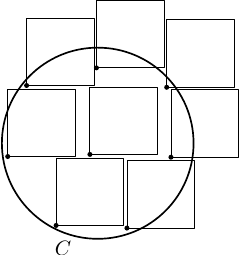}
  \caption{\small Example configuration of eight axis-aligned,
    non-intersecting, unit-squares with their bottom-left corner
    inside a circle $C$ with radius $\sqrt{2}$.}
\label{fig:eight_squares:example}
\end{figure}
Based on Lemma~\ref{lem:packing:squares} we now show that for
any label with anchor $p$ there is no point of time $t \in [0,
1]$ for which there can be more than eight active labels whose anchors
are within distance $\sqrt{2}$ of $p$.
We call a set $X \subseteq \presence$ \emph{conflict-free} if it
contains no pair of presence intervals that are in conflict. Further,
we say that $X$ is in conflict with $I \in
\presence$ if every element of $X$ is in conflict with $I$, and we say
that $X$ contains $t \in [0, 1]$ if every element of $X$ 
contains~$t$.

\begin{lemma}\label{lem:setX}
  For every $t\in [0, 1]$ and every $I\in \presence$ any maximum
  cardinality conflict-free set $X_I(t) \subseteq \presence$ that is in conflict with $I$ and contains $t$ satisfies $|X_I(t)| \le 8$.
\end{lemma}

\begin{proof}
  Assume that there is a time~$t$ and an interval~$I$ such that there
  is a set~$X_I(t)$ that contains more than eight
  intervals. Let~$\ell$ be the label that corresponds to~$I$. For an
  interval~$I'\in X_I(t)$ to be in conflict with~$I$ the anchors of
  the two corresponding labels must have a distance of at
  most~$\sqrt{2}$. Hence, there are~$|X_I(t)|$ labels corresponding to
  the intervals in~$X_I(t)$ with anchors of distance at
  most~$\sqrt{2}$ to the anchor of~$\ell$. By
  Lemma~\ref{lem:packing:squares} we know that two of these labels
  must overlap. This implies that there is a conflict between the
  corresponding intervals contained in~$X_I(t)$, which is a
  contradiction. \qed
\end{proof}
With this lemma we can finally obtain the approximation guarantees
for \GreedyMaxTotal for all activity models.

\begin{theorem}
  Assuming that all labels are unit squares and~$w([a,b])=b-a$,
\GreedyMaxTotal is a 1/24-, 1/16-, 1/8-approximation for AM1--AM3,
respectively, and needs $O(n\log n)$ time for AM1 and
  $O(n^2)$ time for AM2 and AM3.
 \label{thm:maxtotal:apx}
\end{theorem}

\begin{proof}
To show the
  approximation ratios, we consider an arbitrary step of
  \GreedyMaxTotal in which the presence interval~$I = [a,b]_\ell$ is
  selected from~$\presence$. Let $C^I_\ell$ be the set
  of presence intervals in $\presence$ that are in conflict with $I$.
  
  Consider the model AM1. Since~$I$ is the longest interval
  in~$\presence$ when it is chosen, the intervals in $C_\ell^{I}$ must
  be completely contained in~$J=[a-w(I),b+w(I)]$. As~$C_\ell^{I}$
  contains all presence intervals that are in conflict with $I$ it is
  sufficient to consider~$J$ to bound the effect of
  selecting~$I$. Obviously, the interval~$J$ is three times as long
  as~$I$.  By Lemma~\ref{lem:setX} we know that for any $X_I(t)$ it holds that $|X_{I}(t)|\leq 8$ for
  all~$t\in J$. Hence, in an optimal solution there can be at most
  eight active labels at each point~$t\in J$ that are discarded
  when~$[a,b]_\ell$ is selected.  Thus, the cost of
  selecting~$[a,b]_\ell$ is at most~$3\cdot 8\cdot w(I)$.

  For AM2 we apply the same arguments, but restrict the
  interval $J$ to $J =[a,b+w(I)]$, which is only twice as long as~$I$.
  To see that consider for an interval $[c,d]_{\ell'}\in
  C_\ell^{I}$ the prefix~$[c,a]$ if it exists. If $[c,a]$ does not
  exist (because $a<c$), removing~$[c,d]_{\ell'}$ from~$\presence$
  changes~$\presence$ only in the range of~$J$. If~$[c,a]$ exists,
  then again~$\presence$ is only changed in the range of~$I$, because
  by definition~$[c,d]_{\ell'}$ is shortened to an interval that at
  least contains~$[c,a]$ and is still contained in $\presence$. Thus,
  the cost of selecting~$I$ is at most $2 \cdot 8w(I)$.

  Analogously, for AM3 we can argue that it is sufficient to consider
  the interval~$J=[a,b]$. By definition of the update operation of
  \GreedyMaxTotal at least the prefix or suffix subinterval of each
  $[c,d]_{\ell'}\in C_\ell^{I}$ remains in $\presence$ that extends
  beyond~$I$ (if such an interval exists). Thus, selecting~$I$
  influences only the interval~$J$ and its cost is at most $8w(I)$.
  The approximation bounds of $1/24$, $1/16$, and $1/8$ follow
  immediately.

  We use a heap to achieve the time complexity $O(n \log n)$ of
  \GreedyMaxTotal for AM1
  since each interval is inserted and removed
  exactly once.  
  For AM2 and AM3 we use a linear sweep to identify the
  longest interval contained in $\Psi$.  In each step we need $O(n)$
  time to update all intervals in $\Psi$, and we need a total of
  $O(n)$ steps.  Thus, \GreedyMaxTotal needs $O(n^2)$ time in total
  for AM2 and AM3. \qed
\end{proof}

\section{Solving \kRestrictedMaxTotal}
Corollary~\ref{cor:ilp} showed that \kRestrictedMaxTotal can be solved
by integer linear programming in all activity models. In this section
we prove that unlike \GMT the problem \kRestrictedMaxTotal
can actually be solved in polynomial time.  We give a detailed
description of our algorithm for AM1, and then show how it can be
extended to AM2.  Note that solving \kRestrictedMaxTotal is related to
finding a maximum cardinality $k$-colorable subset of $n$ intervals in
interval graphs.  This can be done in polynomial time in both $n$ and
$k$~\cite{Carlisle1995225}. However, we have to consider additional
constraints due to conflicts between labels, which makes our problem
more difficult.  First, we discuss how to solve the case for $k=1$,
then give an algorithm that solves \kRestrictedMaxTotal for $k = 2$,
and extend this result recursively to any constant $k > 2$. Since the
running times of the presented algorithms are, even for small k,
prohibitively expensive in practice, we finally propose an
approximation algorithm for~\kRestrictedMaxTotal.

\subsection{An Algorithm for \TwoRestrictedMaxTotal in AM1}
\label{sub:2restricted}

We start with some definitions before giving the actual algorithm.  We
assume that the intervals of $\presence = \{I_1, \dots, I_n\}$ are
sorted in non-decreasing order by their left endpoints; ties are
broken arbitrarily.  First note that for the case that at most one
label can be active at any given point in time ($k=1$), conflicts
between labels do not matter. Thus, it is sufficient to find an
independent subset of $\presence$ of maximum weight.  This is
equivalent to finding a maximum weight independent set on interval
graphs, which can be done in~$O(n)$ time using dynamic programming
given $n$ sorted intervals~\cite{htc-eafmw2ig-92}. We denote this
algorithm by $\mathcal A_1$. Let $L_1[I_j]$ be the set of intervals
that lie completely to the left of the left endpoint of
$I_j$. Algorithm $\mathcal A_1$ basically computes a table $\T_1$
indexed by the intervals in $\presence$, where an entry $\T_1[I_j]$
stores the value of a maximum weight independent set $Q$ of $L_1[I_j]$
and a pointer to the rightmost interval in $Q$.

We call a pair of presence
intervals~$(I_{i},I_{j})$, %
$i<j$, a \emph{separating pair} if~$I_i$ and $I_j$ overlap and are not
in conflict with each other. Further, a separating pair $\vec v =
(I_p, I_q)$ is \emph{smaller} than another separating pair $\vec w =
(I_i, I_j)$ if and only if $p < i$ or $p = i$ and $q < j$. This
induces a total order and we denote the ordered set of all separating
pairs by $S_2 {=} \{\vec v_1, \ldots, \vec v_z\}$. The weight of a
separating pair $\vec v$ is defined as $w(\vec v) = \sum_{I\in \vec v}
w(I)$.

We observe that a separating pair $\vec v = (I_i, I_j)$ contained in a
solution of \TwoRestrictedMaxTotal splits the set of presence
intervals into two independent subsets. Specifically, a left (right)
subset $L_2[\vec v]$ ($R_2[\vec v]$) that contains only intervals which lie
completely to the left (right) of the intersection of $I_i$ and $I_j$
and are neither in conflict
with $I_i$ nor $I_j$; see Fig.~\ref{fig:example_separating_pair}.

\begin{figure}[tb]
 \centering
 \scalebox{1.07}{\includegraphics[page=1]{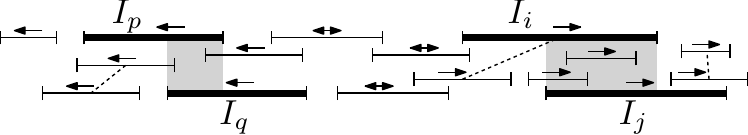}}
 \caption{\small Illustration of presence intervals. Intervals that
   are in
   conflict are connected by a dotted line.  Both $(I_{i}, I_{j})$ and
   $(I_{p}, I_{q})$ are separating pairs.  The intervals of~$L_2[i,j]$
   ($R_2[p,q]$) are marked by a left (right) arrow.}
 \label{fig:example_separating_pair}
\end{figure}

We are now ready to describe our dynamic programming algorithm
$\mathcal A_2$.  For ease of notation we add two dummy separating
pairs to $S_2$. One pair $\vec v_0$ with presence intervals strictly
to the left of $0$ and one pair $\vec v_{z+1}$ with presence intervals
strictly to the right of $1$.  Since all original presence intervals
are completely contained in $[0, 1]$ every optimal solution contains
both dummy separating pairs.  Our algorithm computes a one-dimensional
table $\T_2$, where for each separating pair $\vec v$ there is an
entry $\T_2[\vec v]$ that stores the value of the optimal solution for
$L_2[\vec v]$. We compute $\T_2$ from left to right starting with the
dummy separating pair $\vec v_0$ and initialize $\T_2[\vec v_0] =
0$. Then, we recursively define $\T_2[\vec v_j]$ for every $\vec v_j
\in S_2$ as $\T_2[\vec v_j] = \max_{i < j}\{\T_2[\vec v_i] + w(\vec
v_i) + \mathcal A_1(\vec v_i, \vec v_j) \mid \vec v_i \in S_2,\, \vec
v_i \subseteq L_2[\vec v_j],\, \vec v_j \subseteq R_2[\vec
v_i]\}$. Additionally, we store a backtracking pointer to the
predecessor pair that yields the maximum value. In other words, for
computing $\T_2[\vec v_j]$ we consider all possible direct
predecessors $\vec v_i \in S_2$ with $i<j$, $\vec v_i \cap \vec v_j =
\emptyset$, and no conflict with $\vec v_j$. Each such $\vec v_i$
induces a candidate solution whose value is composed of $\T_2[\vec
v_i]$, $w(\vec v_i)$, and the value of an optimal solution of
algorithm $\mathcal{A}_1$ for the intervals between $\vec v_i$ and
$\vec v_j$ with $\vec v_i$ and $\vec v_j$ active.

Since by construction $L_2[\vec v_{z+1}] = \presence\ \cup\ \vec v_0$,
the optimal solution to $2$-\textsc{Restric\-ted\-MaxTotal} is stored
in $\T_2[\vec v_{z+1}]$ once $\vec v_0$ is removed.  To compute a
single entry $\T_2[\vec v_j]$ our algorithm needs to consider all
possible separating pairs preceding $\vec v_j$, and for each of them
obtain the optimal solution from algorithm $\mathcal A_1$ under some
additional constraints.
For the call $\mathcal A_1(\vec v_i, \vec v_j)$ in the recursive
equation above, we distinguish two cases. If the rightmost endpoint of
$\vec v_i$ is to the left of the leftmost endpoint of $\vec v_j$ then
we run algorithm $\mathcal A_1$ on the set of intervals $L_2[\vec v_j]
\cap R_2[\vec v_i]$ and obtain the value $\mathcal A_1(\vec v_i, \vec
v_j)$. Otherwise, there is an overlap between an interval $I_a$ of
$\vec v_i$ and an interval $I_b$ of $\vec v_j$. Since for $k=2$ no
other interval can cross this overlap, we actually make two calls to
$\mathcal A_1$, once on the set $R_2[\vec v_i] \cap L_2[(I_a,I_b)]$
and once on the set $R_2[(I_a,I_b)] \cap L_2[\vec v_j]$. We add both
values to obtain $\mathcal A_1(\vec v_i, \vec v_j)$.  Since we run
algorithm $\mathcal A_1$ for each of $O(z)$ separating pairs, the time
complexity to compute a single entry of $\T_2$ is $O(nz)$. To compute
the whole table the algorithm repeats this step $O(z)$ times, which
yields a total time complexity of $O(nz^2)$. Note that the number of
separating pairs $z$ is in $O(n^2)$.

We prove the correctness of the algorithm by contradiction. Assume
that there exists an instance for which our algorithm does not compute
an optimal solution and let OPT be an optimal solution. This means,
that there is a smallest separating pair $\vec v_j$ for which the
entry in $\T_2[\vec v_j]$ is less than the value of OPT for $L_2[\vec
v_j]$. Note that $\vec v_j$ cannot be the dummy separating pair $\vec
v_0$ since $\T_2[\vec v_0]$ is trivially correct. Let $\vec v_i$ be
the rightmost separating pair in OPT that precedes $\vec v_j$ and is
disjoint from it (possibly $\vec v_i = \vec v_0$). Since there is no
other disjoint separating pair between $\vec v_i$ and $\vec v_j$ in
OPT, all intervals in OPT between $\vec v_i$ and $\vec v_j$ form a
subset of $R_2[\vec v_i] \cap L_2[\vec v_j]$ that is a valid
configuration for $k=1$. We can obtain an optimal solution for $k=1$
of the intervals in $R_2[\vec v_i] \cap L_2[\vec v_j]$ by computing
$\mathcal A_1(\vec v_i, \vec v_j)$ as described above.  Since, by
assumption,~$\T_2[\vec v_i]$ is optimal,~$\mathcal A_1$ is
correct~\cite{htc-eafmw2ig-92}, and our algorithm explicitly considers
all possible preceding separating pairs including $\vec v_i$, the
entry~$\T_2[\vec v_j]$ must be at least as good as OPT for~$L_2[\vec
v_j]$. This is a contradiction and the correctness of $\mathcal A_2$
follows.

\begin{theorem}\label{thm:A2}
  Algorithm $\mathcal A_2$ solves
  \TwoRestrictedMaxTotal in AM1 in
  $O(nz^2)$ time and $O(z)$ space, where $z$ is the number of
  separating pairs in the input instance.
\end{theorem}

\subsection{An Algorithm for \kRestrictedMaxTotal in AM1}
\label{sub:krestricted}

In the following we extend the dynamic programming algorithm $\mathcal
A_2$ to a general algorithm $\mathcal A_k$ for the case $k>2$. To this
end, we extend the definition of separating pairs to separating
$k$-tuples. A \emph{separating $k$-tuple} $\vec v$ is a set of $k$
presence intervals that are not in conflict with each other and that
have a non-empty intersection $Y_{\vec v} = \bigcap_{I \in \vec v}
I$. We say a separating $k$-tuple $\vec v$ is \emph{smaller} than a
separating $k$-tuple $\vec w$ if $Y_{\vec v}$ begins to the left of
$Y_{\vec w}$. Ties are broken arbitrarily. This lets us define the
ordered set $S_k = \{\vec v_1, \ldots, \vec v_z\}$ of all separating
$k$-tuples of a given set of presence intervals. We say a set $C$ of
presence intervals is \emph{$k$-compatible} if no more than $k$
intervals in $C$ intersect at any point and there are no conflicts in
$C$. Two separating $k$-tuples $\vec v$ and $\vec w$ are
\emph{$k$-compatible} if they are disjoint and $\vec v \cup \vec w$ is
$k$-compatible. The definitions of the sets $R_2[\vec v]$ and
$L_2[\vec v]$ extend naturally to the sets $R_k[\vec v]$ and $L_k[\vec
v]$ of all intervals completely to the right (left) of $Y_{\vec v}$
and not in conflict with any interval in $\vec v$. Now, we recursively
define the algorithm~$\mathcal A_k$ that solves \kRestrictedMaxTotal
given a pair of active $k$-compatible boundary $k$-tuples. Note that
in the recursive definition these boundary tuples may remain
$k$-dimensional even in $\mathcal A_{k'}$ for $k' < k$. For $\mathcal
A_k$ we define as boundary tuples two $k$-compatible dummy separating
$k$-tuples $\vec v_{0}$ and $\vec v_{z + 1}$ with all presence
intervals strictly to the left of $0$ and to the right of $1$,
respectively. The algorithm fills a one-dimensional table $\T_k$.
Similarly to the case $k=2$, each entry $\T_k[\vec v]$ stores the
value of the optimal solution for $L_k[\vec v]$, i.e., the final
solution can again be obtained from $\T_k[\vec v_{z+1}]$. We
initialize $\T_k[\vec v_{0}] = 0$.  Then, the remaining entries of
$\T_k$ can be obtained by computing $\T_k[\vec v_j] = \max_{i <
  j}\{\T_k[\vec v_i] + w(\vec v_i) + \mathcal{A}_{k-1}(\tilde {\vec
  v}_i,\tilde {\vec v}_j) \mid \vec v_i \in S_k,\, \vec v_i \subseteq
L_k[\vec v_j] \cup \vec v_0,\, \vec v_j \subseteq R_k[\vec v_i] \cup
\vec v_{z+1},\, \vec v_0 \cup \vec v_{z+1} \cup \vec v_i \cup \vec v_j
\text{ is } k\text{-compatible}\}$, which uses the algorithm $\mathcal
A_{k-1}$ recursively on a suitable subset of presence intervals
between the boundary tuples $\tilde {\vec v}_i$ and $\tilde {\vec
  v}_j$. Here $\tilde {\vec v}_i$ is defined as the union of the tuple
$\vec v_i$ and all intervals in $\vec v_0 \cup \vec v_{z+1}$ that
intersect the right endpoint of $Y_{\vec v_i}$; analogously $\tilde
{\vec v}_j$ is defined as the union of the tuple $\vec v_j$ and all
intervals in $\vec v_0 \cup \vec v_{z+1}$ that intersect the left
endpoint of $Y_{\vec v_i}$.  This makes sure that in each subinstance
all active intervals that are relevant for that particular subinstance
are known. Note that by the $k$-compatibility condition $\tilde {\vec
  v}_i$ and $\tilde {\vec v}_j$ contain at most $k$ elements each.  In
fact, $\mathcal A_{k-1}(\tilde{\vec v}_i,\tilde{\vec v}_j)$ uses
$\tilde{\vec v}_i$ and $\tilde{\vec v}_j$ as boundary $k$-tuples (and
thus does not create dummy boundary tuples) and the set $R_k[\vec v_i]
\cap L_k[\vec v_j]$ as the set of presence intervals from which
separating $(k-1)$-tuples can be formed.

\newcommand{\thmkrestrictedamonetext}{Algorithm $\mathcal{A}_k$ solves
  \kRestrictedMaxTotal in AM1 in $O(n^{k^2 + k - 1})$ time and
  $O(n^k)$ space.}

\begin{theorem}
\label{thm:krestricted:am1}
\thmkrestrictedamonetext
\end{theorem}
\begin{proof}
We show the correctness of $\mathcal A_k$ by induction on $k$.
Theorem~\ref{thm:A2} shows that the statement is true for $k=2$. Let
$k>2$.
Since $\mathcal A_k$ only considers solutions where adjacent
separating
$k$-tuples are $k$-compatible with each other and the boundary
$k$-tuples, we
cannot produce an invalid solution, i.e., a solution with conflicts or
more
than $k$ active intervals at any point. We prove the correctness by
contradiction. So assume that there is an instance $\presence$ for
which
$\mathcal A_k$ does not compute an optimal solution and let OPT be an
optimal
solution. There must be a smallest separating $k$-tuple $\vec v_j$,
$j>0$, for
which $\T_k[\vec v_j]$ is less than the value of OPT for $L_k[\vec
v_j]$. Let
$\vec v_i$, $i<j$ be the rightmost disjoint separating $k$-tuple in
OPT that
precedes $\vec v_j$ such that the set $\vec v_0 \cup \vec v_i \cup
\vec v_j
\cup \vec v_{z+1}$ is $k$-compatible. By our assumption $\T_k[\vec
v_i]$ has
the same value as OPT on $L_k[\vec v_i]$. For the set of intervals
$L_k[\vec
v_j] \cap R_k[\vec v_i]$ there are at most $k-1$ active intervals at
any point
(otherwise $\vec v_i$ is not rightmost). This means that when we run
algorithm
$\mathcal A_{k-1}$ on that instance with the boundary tuples
$\tilde{\vec
v}_i$ and $\tilde{\vec v}_j$, i.e., $\vec v_i$ and $\vec v_j$ enriched
by all
relevant intervals in $\vec v_0 \cup \vec v_{z+1}$, we obtain by
induction a
solution that is at least as good as the restriction of OPT to that
instance.
Since $\vec v_i$ is a valid predecessor $k$-tuple for $\vec v_j$ the
algorithm
$\mathcal A_{k}$ considers it. So $\T_k[\vec v_j] \ge \T_k[\vec v_i] +
w(\vec
v_i) + \mathcal A_{k-1}(\tilde{\vec v}_i,\tilde{\vec v}_j)$, which is
at least
as good as OPT restricted to $L_k[\vec v_j]$. This is a contradiction
and
proves the correctness.

For proving the time and space complexity let $z_i$ be the number of
separating $i$-tuples in an instance for $1 < i \leq k$. Each $z_i$ is
in $O(n^i)$. We again use induction on $k$. For $\mathcal A_2$
Theorem~\ref{thm:A2} yields $O(n^5)$ time and $O(n^2)$ space, which
match the bounds to be shown. So let $k>2$. The table $\T_k$ has
$O(z_k) \subseteq O(n^k)$ entries and each of the recursive
computations of $\mathcal A_{k-1}$ need $O(n^{k-1})$ space by the
induction hypothesis. Thus the overall space is dominated by $\T_k$
and the bound follows. Checking whether a separating $k$-tuple $\vec
v_i \in S_k$ is a feasible predecessor for a particular $\vec v_j$ can
easily be done in $O(k^2)$ time, which is dominated by the time to
compute $\mathcal A_{k-1}(\tilde{\vec v}_i,\tilde{\vec v}_j)$. So for
the running time we observe that each entry in $\T_k$ makes $O(z_k)$
calls to $\mathcal A_{k-1}$ and hence the overall running time is
indeed $O(n^{2k} \cdot n^{(k-1)^2+(k-1)-1}) = O(n^{k^2 + k -1})$. \qed
\end{proof}

\subsection{Extending the algorithm for \kRestrictedMaxTotal to
  AM2}
\label{sub:extensionAM1and2}

With some modifications and at the expense of another polynomial
factor in the running time we can extend algorithm $\mathcal A_k$ of
the previous section to the activity model AM2, which shows that
\kRestrictedMaxTotal in AM2 can still be solved in polynomial time. In
the following we give a
sketch of the modifications. The important difference between AM1 and
AM2 is that presence intervals can be truncated at their right side if
there is an active conflicting witness label causing the
truncation. We need two modifications to model this behavior. First,
we create for each \emph{original} presence interval $I_i = [a_i,b_i]$
in $\presence$ at most $n$ prefix intervals $I_i^j = [a_i,c_{ij}]$,
where $c_{ij}$ is the start of the first conflict between $I_i$ and
$I_j \in \presence$. Each interval $I_i^j$ inherits the conflicts of
$I_i$ that intersect $I_i^j$. We obtain a modified set of presence
intervals $\presence' = \presence \cup \{I_i^j \mid I_i, I_j \in
\presence \text{ and } I_i, I_j \text{ in conflict}\}$ of size
$O(n^2)$. We create mutual conflicts among all intervals that are
prefixes of the same original interval. This will enforce that at most
one of them is active. We still have to take care that a truncated
interval $I_i^j$ can only be active if $I_j$ (or a prefix of $I_j$) is
active at $c_{ij}$ as a witness.

In order to achieve this we instantiate the algorithm $\mathcal
A_{k'}$ for every $k' \le k$ not only with its two boundary $k$-tuples
$\tilde{\vec v}_0$ and $\tilde{\vec v}_{z+1}$ but also with a set $W$
of at most $k$ witness intervals that are $k$-compatible and must be
made active at some stage of the algorithm. In a valid solution we have $W
\subseteq L_{k'}[\vec v] \cup \vec v$ for the leftmost separating
$k'$-tuple $\vec v$, since otherwise more than $k'$ intervals are
active in $Y_{\vec v}$. However, the truncated intervals in $\vec v$
themselves define a family of $O(n^{k'})$ possible witness sets
$W(\vec v)$ to be respected to the right of $\vec v$. So when we
compute the table entry for a separating $k'$-tuple $\vec v_j$ and
consider a particular predecessor $k'$-tuple $\vec v_i$ we must in
fact iterate over all possible witness sets $W(\vec v_i)$ as well. We
need to make sure that $\vec v_j$ is \emph{$W(\vec v_i)$-compatible},
i.e., $\vec v_j \cup W(\vec v_i)$ is $k$-compatible and $W(\vec v_i)
\subseteq L_{k'}[\vec v_j] \cup \vec v_j$. For the recursive call to
$\mathcal A_{k'-1}(\tilde{\vec v}_i,\tilde{\vec v}_j)$ the initial
witness set $W'$ consists of $W(\vec v_i) \setminus \vec v_j$, i.e.,
those witness intervals of $W(\vec v_i)$ that are not part of $\vec
v_j$.

The increase in running time is caused by dealing with $O(n^2)$
intervals in $\presence'$ and by the fact that instead of one call to
$\mathcal A_{k-1}(\tilde{\vec v}_i,\tilde{\vec v}_j)$ in the
computation of table $\T_k$ we make $O(n^k)$ calls, one for each
possible witness set of ${\vec v}_i$.  By an inductive argument one
can show that the running time is in $O(n^{3k^2 + 2k})$.

\newcommand{\thmkrestrictedamtwotext}{\kRestrictedMaxTotal in AM2 can
be solved in polynomial time.}
\begin{theorem}
  \label{thm:krestricted:am2}
\thmkrestrictedamtwotext
\end{theorem}

It remains open whether \kRestrictedMaxTotal can be solved in
polynomial time in AM3. Another extension of the dynamic programming
algorithm is unlikely, since in AM3 the left and right subinstances
created by a separating $k$-tuple $\vec v$ may have dependencies 
and thus cannot be solved independently any more.
This is because a single
original presence interval $I$ can have subintervals both in $L_k[\vec
v]$ and $R_k[\vec v]$, which cannot simultaneously be active.

\subsection{Approximation of \kRestrictedMaxTotal}

Since the running times of our algorithms for~\kRestrictedMaxTotal
are, even for small $k$,
prohibitively expensive in practice, we propose an approximation algorithm for
\kRestrictedMaxTotal based on \GreedyMaxTotal.

Our algorithm~\kRestrictedGreedy is a simple extension of
\GreedyMaxTotal. Recall that
\GreedyMaxTotal greedily removes the longest interval $I$ from
$\Psi$ and adds it to the set $\Phi$ that contains the active
intervals of the solution. Then, it updates all intervals contained
in $\Psi$ that are in conflict with $I$. This process is repeated
until $\Psi$ is empty. For approximating \kRestrictedMaxTotal we
need to ensure that there is no point in time $t$ that is contained
in more than $k$ intervals in $\Phi$. We call intervals which we
cannot add to $\Phi$ without violating this property \emph{invalid}.

Our modification of \GreedyMaxTotal is as follows. After adding an
interval $I$ to $\Phi$ and handling conflicts as before, we remove
intervals from $\Psi$ that became invalid. We say
that we \emph{ensure that $I$ is valid}. Note that we
cannot shorten those intervals because then we could not ensure that
adding an interval from $\presence$ to $\Phi$ is valid according to
our model.

In order to prove approximation ratios we first introduce the
following lemma that describes the structure of a solution
of~\kRestrictedMaxTotal.

\begin{lemma}
  Let $S$ be a set of intervals such that there is no number that is
  contained in more than $k$ intervals from $S$. Then, there is a
  partition of $S$ into $k$ sets $M_1, \ldots, M_k$, such that no
  two intersecting intervals are in the same set $M_i$.
 \label{lem:krestricted:apx}
\end{lemma}

\newcommand{\I}{\ensuremath{\mathcal{I}}\xspace}
\begin{proof}
  Let $I_1, \ldots, I_m$ be all intervals of $S$ sorted by their left
  endpoints in non-decreasing order. In the following we describe how
  to construct the partition.

  We start with empty sets $M_1, \ldots, M_k$. First, add $I_1$ to
  $M_1$.  Assume that the first $i-1$ intervals have been added to the
  sets $M_1, \ldots, M_k$.  We describe how to add $I_i$. If there is
  an empty set $M_j$ then, we simply add $I_i$ to $M_j$. Otherwise,
  Let $I_{i_1}, \ldots, I_{i_k}$ be the rightmost intervals in the
  sets $M_1, \ldots, M_k$, respectively. We denote the set containing
  those intervals by $R$.  Let $\I = \bigcap_{I \in R} I$. If \I is
  not empty then, due to the order of the intervals, the interval
  $I_i$ cannot begin to the left of $\I$.  It also cannot begin in
  $\I$ because otherwise there would be a number that is contained in
  $k+ 1$ intervals in $S$. Let $M_x, 1 \leq x \leq k$ be the set that
  contains the interval $I \in R$ with leftmost right endpoint among
  the intervals in~$R$. Since
  $I_i$ lies completely to the right of $\I$ it must also lie
  completely to the right of $I$. Thus we can assign $I_i$ to the set
  $M_x$ without introducing intersections. If $\I$ is empty, then
  there must be an interval $I \in R$ with right endpoint to the left
  of another $I' \in R$. Let $M_x, 1 \leq x \leq k$ be the set that
  contains the interval $I$.  Due to the order of the intervals the
  interval $I_i$ lies completely to the right of $I$ and hence we
  assign $I_i$ to $M_x$ without introducing intersections. This
  concludes the proof.
\end{proof}

With this lemma we now can prove the following theorem that makes an
statement about the approximation ratio of~\kRestrictedGreedy  

\newcommand{\thmapxkrestrictedtext}{
  Assuming that all labels are unit squares and~$w([a,b])=b-a$,
\kRestrictedGreedy  is a $1/\min\{3+3k, 27\}$, $1/\min\{3+2k,
19\}$, $1/\min\{3+k, 11\}$-approximation for AM1--AM3,
respectively, and needs $O(n^2)$.
  }
\begin{theorem}
\thmapxkrestrictedtext
\label{thm:krestricted:apx}
\end{theorem}

\newcommand{\optsol}{\ensuremath{\mathcal{L}}\xspace}
\begin{proof}
  We begin by proving its correctness and then we show its
  time complexity.

  Consider  the step in which we add an interval $I = [a, b]$ to
  $\Phi$, and  let $J = [a-w(I),b+w(I)]$.
  Let~$\mathcal L$ be a fixed, but arbitrary optimal solution.
  If $I\in \mathcal L$, there is no lost weight compared to the
  optimal solution when choosing~$I$. 
  
  Thus, assume that~$I\not\in \mathcal L$. Let~$C(I)\subseteq \mathcal
  L$ be the set of intervals that are in conflict with~$I$.
  Identically to the proof of Theorem~\ref{thm:maxtotal:apx} we can
  argue that $w(C(I))=\sum_{I\in C} w(I) \leq (4-X)\cdot 8\cdot w(I)$
  considering activity model
  AM$X$ with~$X\in\{1,2,3\}$.

  We now show that at most $3w(I)$ weight of the optimal solution is
  lost when ensuring that $I$ is valid. 

  By Lemma~\ref{lem:krestricted:apx} we can partition \optsol into $k$
  sets $M_1, \ldots, M_k$ such that no two intersecting intervals are
  in the same set $M_i$.  If $I$ is in \optsol, then we do not lose
  any weight compared to the optimal solution.  Hence, assume that $I$
  is not in \optsol.  Take any $M_i, 1 \leq i \leq k$, remove all
  intervals of $\optsol \setminus \Phi$ that intersect $I$, and add
  $I$ to $M_i$. We denote the set of removed intervals by~$R$. In the
  following we bound the cost of removing the intervals in~$R$. If
  there are intervals in $R$ that are longer than $I$, then we have
  already accounted for them in previous steps. This relies on the
  fact that we consider the intervals sorted by their length in
  non-ascending order, and, hence if those longer intervals are not in
  $\Phi$, we must have removed them in an earlier step. Thus, we only
  need to bound the length of intervals in $R$ which have length at
  most $w(I)$. Those intervals must lie in $J$, and due to the
  definition of $M_i$ they must be disjoint. Hence, the cost is
  bounded by $3w(I)$.

  All together choosing~$I$ causes that at most $3 w(I) + (4-X)\cdot
  8\cdot w(i)$ weight is lost compared to the optimal solution
  considering activity model AM$X$ with~$X\in\{1,2,3\}$.  Finally,
  this yields an approximation factor of 1/27, 1/19, 1/11 for AM1-3,
  respectively.

  For $k < 8$ we can improve $w(C(I))\leq (4-X)\cdot 8\cdot w(i)$ to
  $w(C(I))\leq (4-X)\cdot k\cdot w(i)$ considering activity model
  AM$X$ with~$X\in\{1,2,3\}$ because we know that $I$ cannot be in
  conflict with more than $k$ intervals of the optimal solution.
  Thus, we can bound the loss of choosing $I$ by $3 w(I) + (4-X)\cdot
  k\cdot w(i)$.  In total this yields the claimed approximation ratios
  for the three activity models.
  
  Finally, we argue the correctness of the claimed running time of
  $O(n^2)$. Since the worst-case running time of \GreedyMaxTotal is
  $O(n^2)$ we only need to argue that we can delete those intervals
  from $\Phi$, which are not valid anymore, in $O(n)$ time per
  step. To do this we simply sort the intervals in $\presence$ in
  non-decreasing order by their left-endpoint. We also maintain $\Phi$
  in the same way. Then, we can check for non-valid intervals with a
  simple linear sweep over $\presence$ and $\Phi$. Hence, each
  iteration of the algorithm requires $O(n)$ time, which yields a
  total running time of $O(n^2)$.  \qed
\end{proof}

\section{Conclusions} %
\label{sec:conclusions}
We have introduced a trajectory-based model for dynamic map labeling that satisfies the consistency criteria demanded by Been et al.~\cite{bdy-dl-06}, even in a stronger sense, where each activity change of a label must be explainable to the user by some witness label. Our model transforms the geometric information specified by trajectory, viewport, and labels into the two combinatorial problems \GeneralMaxTotal and \kRestrictedMaxTotal that are expressed in terms of presence and conflict intervals. Thus our algorithms apply to any dynamic labeling problem that can be transformed into such an interval-based problem; the analysis of the approximation ratios, however, requires problem-specific geometric arguments, which must be adjusted accordingly.

We showed that \GeneralMaxTotal is {$\cal NP$}-complete and
$\mathcal W[1]$-hard and presented an ILP model, which we also implemented and evaluated, and constant-factor approximation algorithms for our three different activity models. The problem \kRestrictedMaxTotal, where at most $k$ labels can be visible at any time, can be solved in polynomial time $O(n^{f(k)})$ in activity models AM1 and AM2 for any fixed $k$, where~$f$ is a polynomial function. Due to the $\mathcal W[1]$-hardness of
\GeneralMaxTotal we cannot expect to find better results for the
running times, apart from improving upon the function $f$. We
therefore also presented an $O(n^2)$-time approximation algorithm for
\kRestrictedMaxTotal in all three activity models.

It remains open whether \kRestrictedMaxTotal is polynomially solvable
in activity model AM3. Further, the analysis of the
approximation algorithms for both \kRestrictedMaxTotal and
\GeneralMaxTotal significantly relies on the assumption that
labels are unit squares. Thus, the question arises, whether
constant-factor approximations exist when this assumption is dropped or
softened, e.g., to labels of unit-height.
To answer this question we think that deeper insights into
the structure of conflicts are necessary, e.g., does the geometric
information based on trajectory, viewport and labels imply a useful
structure on the induced label conflict graph?

{
\small
\bibliographystyle{abbrv}
\bibliography{trajectory}

\begin{thebibliography}{10}

\bibitem{aks-lpmir-98}
P.~K. Agarwal, M.~van Kreveld, and S.~Suri.
\newblock Label placement by maximum independent set in rectangles.
\newblock {\em Comput. Geom. Theory \& Appl.}, 11(3-4):209--218, 1998.

\bibitem{bdy-dl-06}
K.~Been, E.~Daiches, and C.~Yap.
\newblock Dynamic map labeling.
\newblock {\em {IEEE} Trans. Visualization and Computer Graphics},
  12(5):773--780, 2006.

\bibitem{bnpw-oarcd-10}
K.~Been, M.~N{\"o}llenburg, S.-H. Poon, and A.~Wolff.
\newblock Optimizing active ranges for consistent dynamic map labeling.
\newblock {\em Comput. Geom. Theory \& Appl.}, 43(3):312--328, 2010.

\bibitem{Carlisle1995225}
M.~C. Carlisle and E.~L. Lloyd.
\newblock On the k-coloring of intervals.
\newblock {\em Discr. Appl. Math.}, 59(3):225--235, 1995.

\bibitem{cc-mir-09}
P.~Chalermsook and J.~Chuzhoy.
\newblock Maximum independent set of rectangles.
\newblock In {\em ACM-SIAM Symp. Discr. Algorithms (SODA'09)}, pages
  892--901, 2009.

\bibitem{fpt-opcpn-81}
R.~J. Fowler, M.~S. Paterson, and S.~L. Tanimoto.
\newblock Optimal packing and covering in the plane are {NP}-complete.
\newblock {\em Inform. Process. Lett.}, 12(3):133--137, 1981.

\bibitem{gnr-clrm-11a}
A.~Gemsa, M.~N{\"o}llenburg, and I.~Rutter.
\newblock Consistent labeling of rotating maps.
\newblock In F.~Dehne, J.~Iacono, and J.-R. Sack, editors, {\em Int.
  Symp. Algorithms \& Data Structures (WADS'11)}, volume 6844 of {\em LNCS}, pages 451--462.
  Springer, 2011.

\bibitem{htc-eafmw2ig-92}
J.~Y. Hsiao, C.~Y. Tang, and R.~S. Chang.
\newblock An efficient algorithm for finding a maximum weight 2-independent set
  on interval graphs.
\newblock {\em Inform. Process. Lett.}, 43(5):229 -- 235, 1992.

\bibitem{Marx05}
D.~Marx.
\newblock Efficient approximation schemes for geometric problems?.
\newblock In G.~Brodal and S.~Leonardi, editors, {\em European
  Symposium on Algorithms (ESA'05)}, volume 3669 of {\em LNCS}, pages 448--459.
  Springer, 2005.

\bibitem{n-cldmust-12}
B.~Niedermann.
\newblock Consistent labeling of dynamic maps using smooth trajectories.
\newblock Master's thesis, Karlsruhe Institute of Technology, June 2012.

\bibitem{sb-cgvsm-04}
M.~Sester and C.~Brenner.
\newblock Continuous generalization for visualization on small mobile devices.
\newblock In P.~F. Fisher, editor, {\em Developments in Spatial Data
  Handling (SDH'04)}, pages 355--368. Springer, 2004.

\bibitem{ww-pla-95}
F.~Wagner and A.~Wolff.
\newblock A practical map labeling algorithm.
\newblock {\em Comput. Geom. Theory \& Appl.}, 7:387--404, 1997.

\bibitem{wwks-trsglp-01}
F.~Wagner, A.~Wolff, V.~Kapoor, and T.~Strijk.
\newblock Three rules suffice for good label placement.
\newblock {\em Algorithmica}, 30:334--349, 2001.

\end{thebibliography}
}
\end{document}